\documentclass{article}

\usepackage{arxiv}

\usepackage[utf8]{inputenc} 
\usepackage[T1]{fontenc}    
\usepackage{hyperref}       
\usepackage{url}            
\usepackage{booktabs}       
\usepackage{amsfonts}       
\usepackage{nicefrac}       
\usepackage{microtype}      
\usepackage{lipsum}
\usepackage{graphicx}
\graphicspath{ {./images/} }

\usepackage{amsmath,amssymb,amsfonts,mathrsfs}
\usepackage{algorithmic}
\usepackage{graphicx}
\usepackage{algorithm,algorithmic}
\usepackage{hyperref}
\hypersetup{hidelinks=true}
\usepackage{textcomp}
\usepackage{amsthm}

\usepackage[numbers,sort&compress]{natbib}

\usepackage{algorithm}

\newtheorem{assumption}{Assumption}
\newtheorem{proposition}{Proposition}
\newtheorem{theorem}{Theorem}
\newtheorem{lemma}{Lemma}

\newtheorem{remark}{Remark}

\usepackage{tikz}
\usetikzlibrary{positioning, arrows.meta, patterns}

\title{Robust Quantum Control via a Model Predictive Control Strategy}

\author{
 Yunyan Lee \\
  School of Engineering\\
  Australian National University\\
  Canberra ACT 2601, Australia \\
  \texttt{Yun-Yan.Lee@anu.edu.au} \\
   \And
Ian R. Petersen \\
  School of Engineering\\
  Australian National University\\
  Canberra ACT 2601, Australia \\
  \texttt{ian.petersen@anu.edu.au} \\
  \And
 Daoyi Dong \\
  School of Engineering\\
  Australian National University\\
  Canberra ACT 2601, Australia \\
  \texttt{daoyi.dong@anu.edu.au} \\
}

\begin{document}
\maketitle
\begin{abstract}
This article presents a robust control strategy using Time-Optimal Model Predictive Control (TOMPC) for a two-level quantum system subject to bounded uncertainties. 
In this method, the control field is optimized over a finite horizon using a nominal quantum system as the reference and then the optimal control for the first time interval is applied and a projective measurement is implemented on the uncertain system. The new control field for the next time interval will be iteratively optimized based on the measurement result. We present theoretical results to guarantee the stability of the TOMPC algorithm. We also characterize the robustness and the convergence rate of the TOMPC strategy for the control of two-level systems.
Numerical simulations further demonstrate that, in the presence of uncertainties, our quantum TOMPC algorithm enhances robustness and steers the state to the desired state with high fidelity. This work contributes to the progress of Model Predictive Control in quantum control and explores its potential in practical applications of quantum technology.
\end{abstract}


\section{Introduction}
\label{sec:introduction}
Quantum control theory has attracted increasing attention in the context of modeling and controlling quantum systems~\cite{dong2010quantum,wiseman2010quantum,altafini2012modeling,dong2022quantum,dalessandro2021introduction}. Quantum control encompasses the process of determining a control field that drives a quantum system to perform a desired operation, such as quantum state preparation~\cite{Ticozzi2013Stabilization, dong2023learning} and synthesis of quantum gates~\cite{kosut2013robust, Dominy2014Characterization, Silva2008Flatness}. Quantum control has demonstrated efficacy in various aspects of quantum systems, including enhancing the fidelity of quantum gates~\cite{zahedinejad2015high, li2006control}, achieving fast state transfer, and implementing quantum error correction~\cite{ahn2002continuous, Pan2017Stabilizing}.

Numerous research directions have emerged within the field of quantum control, such as quantum optimal control~\cite{james2021optimal, albertini2014time, grivopoulos2008optimal}, quantum learning control~\cite{Dong2020, gentile2021learning, wu2016robust, niu2019universal, wu2019learning, chen2014sampling}, quantum feedback control~\cite{Zhang2012Quantum, wiseman_1993quantum, Daeichian2023Feedback, James2008H, Zhang2011Direct, sayrin2001real}, and quantum robust control~\cite{liu2022fault, schirmer2022robust, dong2010robust, protopopescu2003robust, kosut2013robust, Wang2023Quantum, dong2019learning}. 
Quantum optimal control facilitates the development of pulse sequences for guiding quantum systems toward desired states or operations with optimal performance (e.g., minimum time, maximum fidelity, minimum control energy). In quantum learning control, the control design problem of quantum systems is usually formulated as a machine learning task and various learning algorithms (e.g., genetic algorithm, reinforcement learning, differential evolution) have been employed to search for optimal or robust control pulses. Quantum feedback control utilizes information from the output of a quantum plant to design a feedback controller for achieving a desired target. Two main approaches including measurement-based feedback control and coherent feedback control have been developed.
In the development of practical quantum technology, robustness is a key requirement. Various control approaches including feedback control, learning control, $H^\infty$ control \cite{James2008H}, and sliding mode control \cite{dong2023learning} have been developed to enhance the robustness of quantum systems. In this paper,
we focus on enhancing robustness through the use of Model Predictive Control (MPC).

MPC has gained prominence as a control technique, owing to its capacity to control systems with constraints and uncertainties~\cite{grune2017nonlinear, Berberich2021Data, Kohler2021A, Cannon2011Stochastic, Lorenzen2017Constraint}. MPC has also been extended to the realm of quantum control~\cite{clouatre2022model, goldschmidt2022model, hashimoto2017stability, hashimoto2013probabilistic, humaloja2018linear}. 
The authors in \cite{humaloja2018linear} utilized the Cayley-Tustin method to discretize the Schrödinger equation, facilitating an MPC design that assures stabilization. In~\cite{clouatre2022model}, a new tomography algorithm was presented to implement data-driven MPC for quantum systems. The work in \cite{goldschmidt2022model} showed that preparing a desired quantum state with MPC can handle disturbances using quantum tomography. While quantum tomography has been extensively explored~\cite{Wang2018A, qi2013quantum}, obtaining comprehensive information on quantum states via this method is still challenging since the requirement for resources and measurements usually increases exponentially with the system size. Therefore, our study introduces an MPC strategy based on Positive Operator-Valued Measure (POVM)~\cite{nielsen2010quantum} where only a small number of measurements are implemented and the measurements are integrated into developing a quantum MPC algorithm.

This paper aims to present a robust control strategy using MPC for a quantum system subject to bounded uncertainties. 
Our proposed strategy comprises two parts: (i) An open loop optimization problem is solved to generate the control signals, which are subsequently applied to the quantum plant. (ii) a quantum measurement, a POVM in our case, determines the system's state. This measured state is fed back to the open loop optimization problem to determine the control signals for the subsequent cycle. We refer to this process as Quantum Model Predictive Control (qMPC), as visualized in Fig.  \ref{fig: Diagram of qMPC}. Quantum measurement is introduced in this process, and the qMPC strategy offers robustness. In addition to securing the fidelity of the final quantum state, our approach considers the fidelity of the quantum state at each sample time $T_s$ and reduces tracking errors. This approach can be extended to control a broader range of quantum systems, enabling quantum states to follow a predetermined trajectory instead of merely ensuring the fidelity of the final quantum state. 

The main contributions of this paper are summarized as follows: (i) We devise a qMPC strategy by integrating MPC and quantum measurements to achieve robust control of quantum systems. (ii) The stability of this qMPC approach is demonstrated in the two-level quantum system. (iii) We establish a success probability bound associated with the measurement of quantum states using POVM and characterize the convergence rate of the qMPC algorithm. (iv) Numerical results are presented to demonstrate the robustness of the qMPC strategy for quantum control problems.

\begin{figure}[htbp]
    \centering
\begin{tikzpicture}[>=Stealth, node distance=20cm] 

\node[draw, rectangle, align=center] (optimizer) at (0,-3) {optimizer:\\ $\min J_L (|\psi_k\rangle)$ \\for nominal system}; 
\node[draw, rectangle, align=center, above left=2.6cm and -1cm of optimizer] (plant) {quantum plant};
\node[draw, rectangle, align=center, minimum size=0.6cm, right=1.2cm of plant] (POVM) {\qquad};

\draw (POVM.center) ++(0.35,-0.05cm) arc (30:150:0.4cm);
\draw[->] (POVM.center) ++(0,-0.2cm) -- ++(45:0.6cm);

\draw[->] (optimizer.west) -- ++(-2,0) |- node[pos=0.25, left, yshift=30pt] {$u_k^\star$} (plant.west);
\draw[->] (plant.east) -- node[midway, above] {$|\psi_{k+1}\rangle$} (POVM.west);
\draw[double equal sign distance, -Implies] (POVM.east) -- node[midway, above] {$|\psi'_{k+1}\rangle$} ++(2.2,0) |- (optimizer.east);
\draw[dashed, <-] (plant.north) -- ++(0,1) node[midway, right, align=center] {uncertainties}; 

\draw[dashed] (-4,-1.6) -- (4.5,-1.6) node[midway, below , align=right] {\textbf{Offline optimization}} node[midway, above, align=right] {\textbf{Online implementation}};
\end{tikzpicture}
 \caption{
 Diagram of qMPC: In the open loop optimization framework, the cost function \(J_L(|\psi_k\rangle)\) is optimized to generate the input signal \(u_k^\star\). Applying this input signal to the quantum plant, which is affected by uncertainties, yields the quantum state \(|\psi_{k+1}\rangle\). Subsequently, a quantum measurement is conducted to acquire the post-measurement state \(|\psi'_{k+1}\rangle\). This state is subsequently fed back into the optimization process. This process is repeated until the post-measurement state \(|\psi'_{k+1}\rangle\) approaches to the desired target state.
 }
    \label{fig: Diagram of qMPC}
\end{figure}
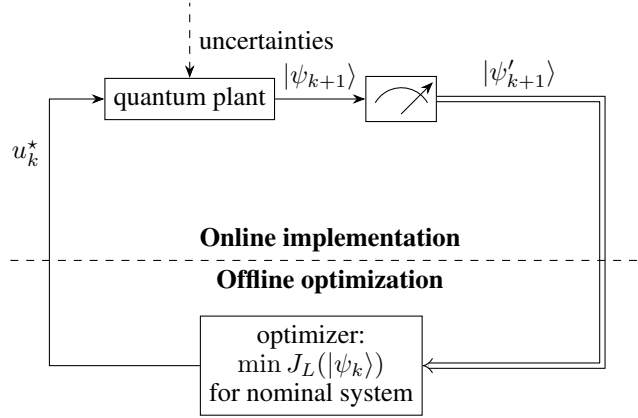
The structure of this paper is as follows. In Section~II, we formulate the quantum MPC problem. Section~III is devoted to presenting stability and robustness guarantees for our approach. In Section~IV, we offer numerical simulations to showcase the error tracking capabilities of our approach in qubit systems, thereby demonstrating the robustness of the proposed method. Finally, we provide concluding remarks in Section~V.

\section{Quantum MPC Problem Formulation}
\label{sec: qMPC problem formulation}
Our discussion is focused on the context of a two-level quantum control system. Firstly, a two-level quantum system is introduced, subsequently leading to a formal definition and exploration of the MPC problem for the nominal system. The last subsection presents a qMPC algorithm specifically designed to control quantum systems with uncertainties.

\subsection{Two-level Quantum Control System}
This subsection introduces a two-level pure-state quantum system (qubit). Its quantum state \(|\psi\rangle\) can be described by a two-dimensional unit vector in a Hilbert space \(\mathcal{H}\). In qubit systems, a useful description is through the Pauli matrices, denoted as:
\begin{equation}
\sigma_x=\left(\begin{array}{cc}
0 & 1 \\
1 & 0
\end{array}\right), \sigma_y=\left(\begin{array}{cc}
0 & -i \\
i & 0
\end{array}\right), \sigma_z=\left(\begin{array}{cc}
1 & 0 \\
0 & -1
\end{array}\right).  
\end{equation}
To accommodate every possible configuration, the free Hamiltonian is chosen as $H_{0} = \sum_{j=x,y,z} r_j \sigma_j$, the control Hamiltonian is given by $H_{u}= \sum_{j=x,y,z} u_j (t) \sigma_j$, and the uncertain Hamiltonian is given by $H_{\Delta} = \sum_{j=x,y,z} \Delta_j (t) \sigma_j$, where $r_j,u_j (t),\Delta_j (t) \in \mathbb{R}$.

By freely choosing a suitable coordinate axis, we can assume $H_{0} = r \sigma_z$. Its eigenstates are denoted by $|0\rangle$ and $|1\rangle$. To facilitate our analysis, we define the vector $\vec{v}(t)=(u_x(t),u_y(t),r+ u_z(t))$ corresponding to the coefficients associated with the free and control Hamiltonians. The vector $\vec{\Delta}(t)=(\Delta_x(t),\Delta_y(t),\Delta_z(t))$, and $\Vec{\sigma}=(\sigma_x,\sigma_y,\sigma_z)$ correspond to uncertainties and the Pauli matrices, respectively. Accordingly, the total Hamiltonian can be expressed as $H_{0} + H_u + H_{\Delta} = \Vec{v}(t)\cdot\vec{\sigma} + \Vec{\Delta}(t)\cdot\vec{\sigma}$, and the dynamics can be described by a Schrödinger equation. Consider the nominal system, which excludes uncertainties:
\begin{align}
\begin{aligned}
i|\dot{\psi}(t)\rangle &=(H_{0} + H_{u})|\psi(t)\rangle \\
&=\vec v(t) \cdot \vec \sigma |\psi(t)\rangle,  
\end{aligned}
\end{align}
where we adopt atomic units with \(\hbar = 1\).

In contrast, the real system, which incorporates uncertainties, is expressed as
\begin{align}
\begin{aligned}
i|\dot{\psi}(t)\rangle &=(H_{0} + H_{u} + H_{\Delta})|\psi(t)\rangle, \\
&=(\vec v(t) + \vec \Delta(t)) \cdot \vec\sigma |\psi(t)\rangle.
\end{aligned}
\end{align}
Now, we assume a fixed sampling time $T_s$. In each sampling period $t \in [kT_s, (k+1)T_s]$ for $k=0,1,2, \ldots$, the uncertainties $\vec{\Delta}_{k}$ are assumed to be time-invariant, unknown but bounded, and the control $\vec v(t)$ is chosen to be constant (i.e., we employ a piece-wise constant control $\vec v_k$).
Discretizing the continuous model with sample period $T_s$ yields the discrete time evolution for both the nominal and real systems as
\begin{align}
\label{eqn: discrete two-level system nominal}
&|\psi_{1|k}\rangle = e^{-i T_s \Vec{v}_{k}\cdot\vec{\sigma}} |\psi_{k}\rangle, \\
&\label{eqn: discrete two-level system}
|\psi_{k+1}\rangle = e^{-i T_s (\vec v_{k} + \vec \Delta_{k}) \cdot \vec\sigma} |\psi_{k}\rangle,    
\end{align}
where \(|\psi_{1|k}\rangle\) denotes the state at the next step in the nominal system, while \(|\psi_{k+1}\rangle\) designates the subsequent state in the uncertain quantum system. 
\subsection{TOMPC for the Nominal System}
Given the discrete-time nominal system as defined in \eqref{eqn: discrete two-level system nominal}, we consider a nonlinear MPC framework. Our research focuses on a time-optimal model predictive control (TOMPC) approach to address controller design for the nominal system. Referring to the proposed method by \cite{verschueren2017stabilizing}, the following optimal control problem (OCP) is considered:
\begin{align}
\label{eqn: general_qMPC}
\begin{aligned}
&\min_{u_0,\ldots,u_{L-1}}  &&J_L\left(|\psi_k\rangle\right) = \sum_{l=0}^{L-1} \theta^{l} \left\| |\bar\psi_{l | k}\rangle - |\psi_{f}\rangle\right\| \\ 
&\text{s.t.} 
&& |\bar\psi_{0 | k}\rangle= |\psi_{k}\rangle, \\
&&& |\bar \psi_{l+1 | k}\rangle = e^{-i \Vec{v}_{l|k}\cdot\vec{\sigma} T_s} |\bar\psi_{l|k}\rangle, \\
&&& |u_{l|k}| \leq B, \\
&&& |\bar\psi_{L | k}\rangle= |\psi_{f}\rangle, 
\end{aligned}
\end{align}
where the cost function \( J_L\left(|\psi_k\rangle\right) \) includes a fixed parameter $\theta \in \mathbb{R}$ and measures the deviation of $|\bar\psi\rangle$ from the desired target state $|\psi_f\rangle$. In our study, the distance between quantum states is defined as the trace distance 
$ \left\| |\bar\psi_{l | k}\rangle - |\psi_{f}\rangle\right\| = \sqrt{1-|\langle \psi_f|\bar\psi_{l | k}\rangle|^2}$ \cite{gilchrist2005distance}. The trace distance is equal to half the Euclidean norm of the difference between the corresponding Bloch vectors. Additionally, the inequality constraints ensure that the magnitude of the control fields $u_{l|k}$ remains bounded by $B$, where $B \in \mathbb{R}$.

Upon finding the optimal value of the cost function \( J_L\left(|\psi_k\rangle\right) \), the optimal control sequence over the prediction horizon \( L \) is given as \( \{u^\star_{0 | k}, u^\star_{1 | k}, \ldots, u^\star_{L-1 | k}\} \), which corresponds to a set of vectors \(\{\vec v^\star_{0 | k}, \vec v^\star_{1 | k}, \ldots, \vec v^\star_{L-1 | k}\}\). Following \eqref{eqn: discrete two-level system nominal}, a unitary matrix is calculated as \( U^\star_{l|k} = e^{-i T_s \vec{v}^\star_{l|k} \cdot \vec{\sigma}} \). In the nominal system, the state \( |\psi_{l | k}\rangle \) at step \( k+l \) is then expressed as 
\begin{align}
    |\psi_{l|k}\rangle = \prod_{j=0}^{l-1} U^\star_{l-1-j|k} |\psi_k\rangle.
\end{align}


\subsection{qMPC for an uncertain quantum system}
Considering the quantum system with uncertainties \eqref{eqn: discrete two-level system}, our objective is to stabilize the quantum state to the target state $|\psi_f\rangle$ despite the presence of uncertainties $H_{\Delta}$. Here, we presume that the initial state is $|\psi_0\rangle$, and the control fields $u_{l|k}$ are piece-wise constant satisfying $|u_{l|k}| \leq B$. 
Given that we optimize the quantum state via the optimization problem \eqref{eqn: general_qMPC}, this process is referred to as quantum TOMPC (qTOMPC) and is outlined in Algorithm \ref{alg:qMPC}.

Algorithm \ref{alg:qMPC} commences by taking the initial state $|\psi_0\rangle$ and target state $|\psi_f\rangle$ of a quantum system, alongside a prediction horizon $L$. The next step involves solving OCP \eqref{eqn: general_qMPC}, such that $\{u^\star_{0 | k}, u^\star_{1 | k}, \ldots, u^\star_{L-1 | k}\}  = \arg\min_{\mathcal{U}} J_L (|\psi_k\rangle)$ is obtained. Following this, the algorithm uses the initial input signal $u^\star_{0|k}$ and predicts the upcoming quantum state $|\psi_{1|k}\rangle$ for the nominal system. Based on the state $|\psi_{1|k}\rangle$, a POVM set $M_{1|k}$ is constructed. Consequently, we measure the quantum state according to the POVM $M_{1|k}$, thus obtaining the post-measurement state, $|\psi_{k+1}'\rangle$. Then, the state \( |\psi_k\rangle \) updates to \( |\psi'_{k+1}\rangle \). This process is continued for a total of \( N \) steps.

\begin{algorithm}[htbp]
\caption{qTOMPC: TOMPC for quantum systems}\label{alg:qMPC}
\begin{algorithmic}[1]
\REQUIRE Initial state $|\psi_0\rangle$ and target state $|\psi_f\rangle$, prediction horizon $L$, performance index $J_L(|\psi_k\rangle)$
\STATE \textbf{repeat} (for each step $k=0,1, \ldots, N-1$)
\STATE Calculate optimal control $\{u^\star_{0|k}, u^\star_{1|k}, \ldots, u^\star_{L-1|k}\}$ by minimizing $J_L(|\psi_k\rangle)$ using the nominal model
\STATE Use $u^\star_{0|k}$ to get $|\psi_{1|k}\rangle$ from the nominal model
\STATE Make a measurement with $M_{1|k}=$ $\left\{\left|\psi_{1|k}\right\rangle\left\langle\psi_{1|k}|, I-| \psi_{1|k}\right\rangle\left\langle\psi_{1|k}\right|\right\}$, and obtain the post-measurement state $\left|\psi'_{k+1}\right\rangle$
\STATE \textbf{update} $|\psi_k\rangle=\left|\psi'_{k+1}\right\rangle$
\end{algorithmic}
\end{algorithm}


Algorithm \ref{alg:qMPC} provides an MPC-based control scheme for uncertain quantum systems that reschedules the quantum trajectory based on measurement result, i.e., post-measurement states. This approach allows for the stabilization of the system state even in the presence of uncertainties.

\section{Theoretical Analysis of Stability and Robustness}
We now examine the stability and robustness properties of the qTOMPC scheme, as described in Algorithm \ref{alg:qMPC}. This analysis is split across two subsections: the stability of the nominal system, and the robustness of a two-level quantum system with uncertainties.

\subsection{Stability Guarantee of the Nominal System}
Using the framework provided by \cite{verschueren2017stabilizing}, we demonstrate OCP \eqref{eqn: general_qMPC} leads to a time-optimal control. Based on their work, we further establish asymptotic stability with our nominal system.

We will show that the OCP \eqref{eqn: general_qMPC} can be equated to the following formulation:
\begin{align}
\label{eqn: TOMPC}
\begin{aligned}
&\min_{L,u_0,\ldots,u_{L-1}}  && L \\ 
&\text{s.t.} 
&& |\bar\psi_{0 | k}\rangle= |\psi_{k}\rangle, \\
&&& |\bar \psi_{l+1 | k}\rangle = e^{-i \Vec{v}_{l|k}\cdot\vec{\sigma} T_s} |\bar\psi_{l|k}\rangle, \\
&&& |u_{l|k}| \leq B \\
&&& |\bar\psi_{L | k}\rangle= |\psi_{f}\rangle, 
\end{aligned}
\end{align}
which is referred to be as the time-optimal control problem.
To equate OCP \eqref{eqn: general_qMPC} and \eqref{eqn: TOMPC}, the following assumptions are required.
\begin{assumption}
\label{ass: target stable}
    The target state $|\psi_f\rangle$ is an eigenstate of the free Hamiltonian $H_0$, i.e., with a control signal $u^\star_{0|k}=0$, we can write $|\psi_f\rangle=e^{-i H_0}|\psi_f\rangle$ in \eqref{eqn: discrete two-level system nominal}.
\end{assumption}
\begin{assumption}
\label{ass: enough horizon}
    Problem \eqref{eqn: TOMPC} is feasible, given $L \geq L^{\star}$, where $L^{\star}$ denotes the minimal horizon required to achieve the desired target quantum state.
\end{assumption}
With Assumptions \ref{ass: target stable} and \ref{ass: enough horizon}, we show that OCP \eqref{eqn: general_qMPC} can be cast as a time-optimal control problem in Theorem \ref{Thm: TOMPC}.
\begin{theorem}[Equivalent Representation]
\label{Thm: TOMPC}
    If a two-level quantum control problem satisfies Assumptions \ref{ass: target stable} and \ref{ass: enough horizon}, there exists \(\theta_1\) such that for all \(\theta \geq \theta_1\), the solution to \eqref{eqn: general_qMPC} satisfies \(|\psi_{l|k}\rangle=|\psi_{f}\rangle\) for \(l=L^\star, \ldots, L\).
\end{theorem}

This proof mainly follows \cite{verschueren2017stabilizing}, but we modify some details to fit a two-level quantum system as shown in Appendix \ref{sec: proof theorem TOMPC}.

In Appendix \ref{sec: proof theorem TOMPC}, it is pointed out that even if constructing the cost function with the Euclidean norm, the optimization problem is still equivalent to a time-optimal control problem.
\begin{remark}
     Based on \cite{wright2006numerical}, consider the penalty function $\mu \left|x^{L^\star}-x_f \right|_p$, where $\mu \geq \left| \lambda \right|_D$, $\mu \in \mathbb{R}$, and $p$ represents an arbitrary norm. Here, $D$ denotes the dual norm corresponding to $p$. Given these conditions, the penalty function $\mu \left|x^{L^\star}-x_f \right|_p$ ensures that the problem exactly satisfies the constraint. However, this approach is inapplicable to a quadratic form, which does not satisfy the constraints exactly.
\end{remark}

In our context, the term ‘TOMPC’ refers to a nominal system which is controlled via Algorithm \ref{alg:qMPC} in the absence of measurements. That is, TOMPC only includes Steps 1 to 3 and Step 5 of Algorithm \ref{alg:qMPC}. In the TOMPC process, \( u^\star_{0|k} \) serves to determine \( |\psi_{1|k}\rangle \), which subsequently acts as the initial state for successive iterations. We term the control of the nominal system as TOMPC through our paper, and when uncertainties and quantum measurement are introduced, we refer to the whole process as qTOMPC.

Given Theorem \ref{Thm: TOMPC}, we can establish the asymptotic stability of the TOMPC process for a two-level quantum system in light of Proposition \ref{prop: bound}. 
\begin{proposition}
\label{prop: bound}
Let \( J^\star_L\left(|\psi_k\rangle\right) \) represent the minimal value of \( J_L\left(|\psi_k\rangle\right) \). Then, 
there exists a \(\mathscr{K}_\infty\) function \( \alpha (\cdot) \) such that 
\begin{align}
    J^\star_L\left(|\psi_k\rangle\right) \leq \alpha (\left\| |\psi_k\rangle - |\psi_f\rangle \right\|) . 
\end{align}
Here, the function \( \alpha (\cdot) \) belongs to \(\mathscr{K}_\infty\), that is, \( \alpha (0)=0 \) and is strictly increasing and unbounded.
\end{proposition}
\begin{proof}
Given that the system under consideration is a two-level quantum system, the following relationships hold:
\begin{align}
\begin{split}
        J^\star_L\left(|\psi_k\rangle\right) &\leq J_L\left(|\psi_k\rangle\right) \\
        &=\sum_{l=0}^{L-1} \theta^{l} \left\| |\bar\psi_{l | k}\rangle - |\psi_{f}\rangle\right\| \\
        &\leq \frac{\theta^{L}-1}{\theta-1},
\end{split} 
\end{align}
where \( L \) and \( \theta \) are finite positive values. Therefore, we can find a function \( \alpha (\cdot) \) such that 
\begin{align}
J^\star_L\left(|\psi_k\rangle\right) \leq \alpha (\left\| |\psi_k\rangle - |\psi_f\rangle \right\| ).
\end{align}
\end{proof}
Drawing upon Proposition \ref{prop: bound} and under the assumption that \( \theta >1 \), we can utilize the framework described in \cite[Theorem 2]{verschueren2017stabilizing} to establish the asymptotic stability for qTOMPC with the time-optimal control problem \eqref{eqn: general_qMPC}. Notably, a modification from the 1-norm used in \cite[Theorem 2]{verschueren2017stabilizing} to the Euclidean norm is required.

\subsection{Stability Guarantee Considering Uncertainties and Measurement}
We demonstrate the robustness of the qTOMPC strategy in a two-level quantum system with an uncertain Hamiltonian.
Theorem \ref{prop: Minimum Probability in a Two-Level Closed System-general case}, asserting that, with a fixed sampling time and a bound on the uncertainties, the minimum probability of a correct measurement for each step in a two-level quantum system can be determined. Theorem \ref{thm: exponentially stable} points out that a two-level quantum system controlled by the qTOMPC strategy will approach the target state with high probability despite uncertainties. Furthermore, we present a lower bound on the convergence rate to achieve the desired state based on Theorem \ref{cor:convergence_rate}.

\begin{theorem}[Minimum Success Probability]
\label{prop: Minimum Probability in a Two-Level Closed System-general case}
Consider a two-level quantum system \eqref{eqn: discrete two-level system}. When the sampling time $T_s$ satisfies the condition $T_s < \pi/2\bar{\Delta}$ where $\bar{\Delta}$ denotes the uncertain bound such that $|\Vec{\Delta}|\leq\bar{\Delta}$, then, the probability of transferring to the correct nominal state $|\psi_{1|k}\rangle$ satisfies 
\begin{align}
 |\langle \psi_{1|k}|\psi_{k+1}\rangle|^2 \geq \cos^2(\bar{\Delta} T_s).
\end{align}
$\quad$

\end{theorem}

\begin{proof}
We consider the projection of the actual state $|\psi_{k+1}\rangle$ onto the reference state $|\psi_{1|k}\rangle$ and calculate the associated probability as follows:
\begin{align}
\label{eqn: deviation probability}
 |\langle \psi_{1|k} | \psi_{k+1}\rangle|^2&= |\langle \psi_{k} | e^{i T_s (\vec v \cdot \vec \sigma+\vec \Delta \cdot \vec \sigma)} e^{-i T_s (\vec v \cdot \vec\sigma)} | \psi_{k}\rangle |^2.
\end{align}

Given a vector $\vec\omega$ represented as $\vec\omega= |\vec\omega| \frac{\vec \omega}{|\vec\omega|}$, we introduce the following relation \cite{sakurai1995modern}:
\begin{align}
\label{eqn:Pauli relation}
e^{-i|\vec\omega| (\frac{\vec \omega}{|\vec\omega|} \cdot \vec \sigma)}= \cos(|\vec\omega|) \mathbb{I} -i (\frac{\vec \omega}{|\vec\omega|} \cdot \vec \sigma) \sin(|\vec\omega|).
\end{align}
Based on the properties of the exponential of a Pauli matrix, we have
\begin{align}
e^{-i (\vec v \cdot \vec\sigma) T_s}= \cos(T_s |\vec v|) \mathbb{I} -i \frac{\vec v \cdot \vec\sigma}{|\vec v|}  \sin(T_s |\vec v|),\\
\begin{aligned}
e^{-i(\vec v \cdot \vec\sigma+\vec \Delta \cdot \vec\sigma) T_s}&= \cos(T_s |\vec v + \vec \Delta|) \mathbb{I} \\
&\phantom{=}-i \frac{(\vec v + \vec \Delta) \cdot \vec\sigma}{|\vec v + \vec \Delta|}  \sin(T_s |\vec v + \vec \Delta|).
\end{aligned}
\end{align}
Therefore, the probability \eqref{eqn: deviation probability} can be written in the following form:
\begin{align}
\begin{aligned}
 |\langle \psi_{1|k} | \psi_{k+1}\rangle|^2=& |\langle \psi_k | [\cos(T_s |\vec v|) \mathbb{I} +i \frac{\vec v \cdot \vec\sigma}{|\vec v|} \sin(|\vec v| T_s)] \\
&\cdot  [\cos(T_s |\vec v + \vec \Delta|) \mathbb{I} \\
&\quad-i \frac{(\vec v + \vec \Delta) \cdot \vec\sigma}{|\vec v + \vec \Delta|}  \sin(|\vec v + \vec \Delta| T_s)] |\psi_k\rangle|^2.
\end{aligned}
\end{align}

Expressing the quantum state $|\psi_k\rangle$ in terms of the Bloch vector $\vec n = (x, y, z)$,
the expectation value can be written as:
\begin{align}
\label{eqn: exp expansion}
\begin{aligned}
\langle \psi_{1|k} | \psi_{k+1} \rangle &= \cos(T_s |\vec v|) \cos(T_s |\vec v + \vec \Delta|) \\
&\phantom{=}+ i \frac{\vec v \cdot \vec n}{|\vec v|} \sin(T_s |\vec v|) \cos(T_s |\vec v + \vec \Delta|)\\
&\phantom{=}- i \frac{(\vec v + \vec \Delta) \cdot \vec n}{|\vec v + \vec \Delta|} \cos(T_s |\vec v|) \sin(T_s |\vec v + \vec \Delta|)  \\
&\phantom{=}+ \frac{\vec v \cdot (\vec v + \vec \Delta)}{|\vec v| |\vec v + \vec \Delta|} \sin(T_s |\vec v|) \sin(T_s |\vec v + \vec \Delta|))\\
&\phantom{=}+ i \frac{\vec v \times (\vec v + \vec \Delta) \cdot \vec n}{|\vec v| |\vec v + \vec \Delta|} \sin(T_s |\vec v|) \sin(T_s |\vec v + \vec \Delta|)),
\end{aligned}
\end{align}
where $\times$ denotes the cross-product operation for vectors. We want to establish a lower bound for $|\langle \psi_{1|k} | \psi_{k+1}\rangle|^2$. We utilize the fact that $|a+bi|^2 \geq |a|^2$ for $a,b \in \mathbb{R}$ and have
\begin{align}
\begin{aligned}
|\langle \psi_{1|k} | \psi_{k+1} \rangle|^2 &\geq |\cos(T_s |\vec v|) \cos(T_s |\vec v + \vec \Delta|) \\
&\phantom{=}+\frac{\vec v \cdot (\vec v + \vec \Delta)}{|\vec v| |\vec v + \vec \Delta|} \sin(T_s |\vec v|) \sin(T_s |\vec v + \vec \Delta|))|^2.
\end{aligned}
\end{align}
It is worth mentioning that the real part in \eqref{eqn: exp expansion} is independent of the values of $(x, y, z)$.


In the special situation where the uncertain vector $\vec{\Delta}$ is parallel to the vector $\vec{v}$, it becomes feasible to multiply the matrices on exponents directly, given the commutation, \( [\vec{v} \cdot \vec{\sigma}, (\vec{v} + \vec{\Delta}) \cdot \vec{\sigma} ]=0 \) in \eqref{eqn: deviation probability}. The expectation can be directly calculated when the condition \( \vec{v} \parallel \vec{\Delta} \) is satisfied as follows:
\begin{align}
\begin{aligned}
    \langle \psi_{1|k} | \psi_{k+1}\rangle &= \langle \psi_{k} | e^{i T_s (\vec{v}+\vec{\Delta}) \cdot \vec{\sigma}} e^{-i T_s \vec{v}\cdot \vec{\sigma}} | \psi_{k}\rangle\\
    &= \langle \psi_{k} | e^{i T_s (\vec{\Delta} \cdot \vec{\sigma})} | \psi_{k}\rangle . 
\end{aligned}
\end{align}
As a result, we can deduce $ \Re{\langle \psi_{1|k} | \psi_{k+1}\rangle}=\cos(T_s |\vec{\Delta}|)$.

In the sampling interval, $T_s < \pi/2 |\vec \Delta|$, we evaluate the difference between the real part of expectation value $\langle \psi_{1|k} | \psi_{k+1} \rangle$ and the value corresponding to the parallel condition. This difference is defined as:
\begin{align}
\label{eqn: abstract inequality}
\begin{aligned}
h(T_s)=&\cos(T_s |\vec v|) \cos(T_s |\vec v + \vec \Delta|) \\
&+ \frac{(\vec v \cdot (\vec v + \vec \Delta))}{|\vec v| |\vec v + \vec \Delta|}\sin(T_s |\vec v|) \sin(T_s |\vec v + \vec \Delta|)) \\
&- \cos(|\vec \Delta| T_s).
\end{aligned}
\end{align}
Using identities relating to the sum and difference of two products  yields:
\begin{align}
\label{eqn: abstract inequality sum difference products}
\begin{aligned}
     h(T_s)&=\frac{|\vec{v}| |\vec{v}+\vec{\Delta}|+\vec{v}\cdot(\vec{v}+\vec{\Delta})}{2|\vec{v}| |\vec{v}+\vec{\Delta}|} \cos(T_s (|\vec{v}+\vec{\Delta}|-|\vec{v}|)) \\
     &\phantom{=} +\frac{|\vec{v}| |\vec{v}+\vec{\Delta}|-\vec{v}\cdot(\vec{v}+\vec{\Delta})}{2|\vec{v}| |\vec{v}+\vec{\Delta}|} \cos(T_s (|\vec{v}+\vec{\Delta}|+|\vec{v}|))\\
 &\phantom{=}-\cos(|\vec\Delta| T_s).
\end{aligned}
\end{align}
To evaluate the difference $h(T_s)$, we consider the following two distinct cases.

\textbf{Case 1:} $0 < T_s \leq \frac{\pi}{|\vec v + \vec \Delta|+|\vec \Delta|}$ and $\vec v \nparallel \vec \Delta$.



The proof of this case is shown in Appendix \ref{sec: Proof of Decreasing Series in general case}.  



\textbf{Case 2:} $\frac{\pi}{|\vec v+\vec \Delta|+|\vec \Delta|} < T_s < \frac{\pi}{2|\vec \Delta|}$ and $\vec v \nparallel \vec \Delta$.

To evaluate $h(T_s)$, we note that when $T_s=\frac{\pi}{|\vec v+\vec \Delta|+|\vec \Delta|}$, the second term on the right hand side of \eqref{eqn: abstract inequality sum difference products} reaches its minimum. Additionally, the sum of the first and third terms is a monotonically increasing function. This assertion can be verified from its differential and employing the conditions \eqref{eqn: amp inequality} and \eqref{eqn: freq inequality} in Appendix \ref{sec: Proof of Decreasing Series in general case}. Therefore, in Case 2, $h(T_s)>0$. 


Through the above analysis, we conclude that the condition \( h(T_s)\geq 0 \) always holds. Consequently, we infer that when \( T_s < \frac{\pi}{2}|\vec{\Delta}| \), \( |\langle \psi_{1|k}| \psi_{k+1}\rangle| \geq \cos (T_s |\vec \Delta|) \). Furthermore, considering \( |\vec{\Delta}|\leq\bar{\Delta} \), it follows that:
\begin{align}
    |\langle \psi_{1|k} | \psi_{k+1} \rangle|^2 \geq \cos^2(\bar{\Delta} T_s).
\end{align}
\end{proof}


Theorem \ref{Thm: TOMPC} presents that the system reaches the target state $|\psi_f\rangle$ by the time optimal step $L^\star$ and Theorem \ref{prop: Minimum Probability in a Two-Level Closed System-general case} gives the minimum probability of the system to remain in the state corresponding to the nominal trajectory. This allows us to establish Proposition \ref{thm: exponentially stable} to verify the robustness of qTOMPC.

\begin{proposition}[Probability Iteration Representation]
\label{thm: exponentially stable}
Given the qTOMPC strategy for the two-level system in Theorem \ref{prop: Minimum Probability in a Two-Level Closed System-general case}, let $P_{\textnormal{tar}}(N)$ denote the probability of the system reaching the target state by Step $N$. A lower bound for the probability $P_{\textnormal{tar}}$ is 
\begin{align}
    P_{\textnormal{tar}} (N) \geq 1- [\sin^2 (\bar\Delta T_s)] \sum_{l=1}^{L} [\cos^2 (\bar\Delta T_s)]^{l-1} F_{N-l} ,
\end{align}
where $F_k$ represents the probability of the system failing to achieve the target state by the $k$-th step.
\end{proposition}

\begin{proof}
This problem is analogous to flipping a biased coin. In this interpretation, a ‘head’ (represented as $H$) corresponds to the system remaining on the nominal trajectory, while obtaining a ‘tail’ (represented as $T$) corresponds to a deviation from the nominal trajectory. Additionally, we let $X$ to be an event that can be either $H$ or $T$. Therefore, when we toss a $H$ for $L$ times in a row, the system must evolve to the desired target state since $L > L^\star$ according to Assumption \ref{ass: enough horizon}, and $L^\star$ is the time optimal step number required to reach the target for the nominal system.

Evaluating the probability of failure to remain on the nominal trajectory, $F_k$, we can observe $F_k=1$ for $k=1,\ldots, L-1$, and $F_L=1-[\cos^2 (\bar\Delta T_s)]^L$. 

To calculate $F_{L+1}$, we consider the following hypothetical coin-toss sequences:
\begin{align}
\label{eqn: hypothetical coin-toss sequences}
\begin{aligned}
      &XXX \ldots XXX T\\
    &XXX  \ldots XXTH\\
    &\phantom{XXX..} \vdots \\
    &XXTHH \dots HH\\
    &XTHHH \dots HH.
\end{aligned}
\end{align}
The probability, $F_{L+1}$, is the sum of the probabilities corresponding to this coin-toss sequence:
\begin{align}
\label{eqn: iter L+1}
    F_{L+1} =  [\sin^2 (\bar\Delta T_s)] \sum_{l=1}^{L} [\cos^2 (\bar\Delta T_s)]^{l-1} F_{L+1-l}.
\end{align}
This follows from the fact that the probability corresponding to obtaining a $T$ is $[\sin^2 (\bar\Delta T_s)]$, whereas the probability corresponding to obtaining a $H$ is $[\cos^2 (\bar\Delta T_s)]$.

For any step $N$, the probability $F_N$ only depends on $F_{N-L},\ldots, F_{N-1}$ in the same hypothetical coin-toss sequences \eqref{eqn: hypothetical coin-toss sequences}. That is, 
\begin{align}
\label{eqn: iter N}
    F_N =  [\sin^2 (\bar\Delta T_s)] \sum_{l=1}^{L} [\cos^2 (\bar\Delta T_s)]^{l-1} F_{N-l}.
\end{align}
In light of this, upon evaluating the probability of the system reaching the target state by the $N$-th step, we derive:
\begin{align}
    P_{\textnormal{tar}} (N) \geq 1- [\sin^2 (\bar\Delta T_s)] \sum_{l=1}^{L} [\cos^2 (\bar\Delta T_s)]^{l-1} F_{N-l}.
\end{align}
\end{proof}
This proposition establishes an iteration for the probability $F_k$ and provides a formula of the probability of achieving the target state $P_{\text{tar}}(N)$. We can further construct an upper bound convergence rate for $F_k$, which equivalently signifies the lower bound convergence rate for $P_{\text{tar}}(N)$. To derive this bound, we employ the following two lemmas: Rouché's theorem as described in \cite{bak2010complex} and Descartes's rule of signs presented in \cite{Xiaoshen2004A}.

\begin{lemma}[Rouché's Theorem \cite{bak2010complex}]
Let $f$ and $g$ be analytic functions both inside and on a closed, regular curve $\gamma$. If for every $z \in \gamma$, $|f(z)| > |g(z)|$, then it follows that:
\begin{align}
    \mathbb{Z}(f+g) = \mathbb{Z}(f) \quad \text{within the confines of } \gamma.
\end{align}
Here, $\mathbb{Z}$ denotes the number of roots inside the curve $\gamma$.  
\end{lemma}

\begin{lemma}[Descartes's Rule of Signs \cite{Xiaoshen2004A}]
Given a polynomial \( p(x) = a_0 x^{b_0} + \dots + a_n x^{b_n} \) with non-zero real coefficients \( a_i \) and integer exponents \( b_i \) such that \( 0 \leq b_0 < b_1 < \dots < b_n \), the number of positive real roots (with multiplicities) of \( p(x) \) is either equal to or less than an even number, the number of sign changes in the sequence \( a_0, \dots, a_n \). The same holds for negative roots considering the coefficients of \( p(-x) \).
\end{lemma}

Given Proposition \ref{thm: exponentially stable}, we derive an expression for the convergence rate of \(F_N\), defined as \cite{ortega2000iterative}:
\begin{align}
    \eta = \sup \lim_{N \to \infty} (F_N)^{\frac{1}{N}}.
\end{align}
\begin{theorem}[Convergence Rate]
\label{cor:convergence_rate}
An expression for $\eta$ is given in the following three cases:
\begin{enumerate}
    \item When \(\cos^2 (\bar\Delta T_s) < \frac{L}{L+1}\), 
    \begin{align}
    \eta = \min\left\{ 1-\alpha, \frac{2L}{L+1} - \cos^2 (\bar\Delta T_s) \right\}.
    \end{align}
    Here, $\alpha=\sin^2 (\bar\Delta T_s) [\cos^2 (\bar\Delta T_s)]^L$.  
    \item When \(\cos^2 (\bar\Delta T_s) > \frac{L}{L+1}\), 
    \begin{align}
    \eta = \frac{2L}{L+1} - \cos^2 (\bar\Delta T_s).
    \end{align}
    \item When \(\cos^2 (\bar\Delta T_s) = \frac{L}{L+1}\), 
    \begin{align}
    \eta = \cos^2 (\bar\Delta T_s).
    \end{align}
\end{enumerate}
\end{theorem}

\begin{proof}
By applying the z-transform to \eqref{eqn: iter N}, we obtain
\begin{align}
\label{eqn: original z trans}
    F_N (z) = [\sin^2 (\bar\Delta T_s)] F_N (z) \sum_{l=1}^{L} [\cos^2 (\bar\Delta T_s)]^{l-1} z^{-l}.
\end{align}
To simplify the expression for $F_N (z)$, we multiply both sides of \eqref{eqn: original z trans} by $[z-\cos^2 (\bar\Delta T_s)]\cdot z^L$, which yields
\begin{align}
\begin{aligned}
    &z^{L+1} F_N (z) - z^{L} \cos^2 (\bar\Delta T_s) F_N (z)  \\
    =& [\sin^2 (\bar\Delta T_s)] F_N (z) [z-\cos^2 (\bar\Delta T_s)] \sum_{l=1}^{L} [\cos^2 (\bar\Delta T_s)]^{l-1} z^{L-l}\\
    =&[\sin^2 (\bar\Delta T_s)] F_N (z) (z^L-[\cos^2 (\bar\Delta T_s)]^L).
\end{aligned}
\end{align}
This can be rearranged and simplified as
\begin{align}
    z^{L+1} F_N (z) - z^L F_N (z) = -\alpha F_{N} (z).
\end{align}
As a result, we obtain the characteristic polynomial,  
\begin{align}
\label{eqn: z_polynomial}
    g(z)=z^{L+1} - z^L + \sin^2 (\bar\Delta T_s)[\cos^2 (\bar\Delta T_s)]^L = 0.
\end{align}
Since we multiply \eqref{eqn: original z trans} by the term $[z-\cos^2 (\bar\Delta T_s)]$, \(z = \cos^2 (\bar\Delta T_s)\) is an apparent solution of $g(z)=0$, and we label \(z_1 = \cos^2 (\bar\Delta T_s)\) as a root of \eqref{eqn: z_polynomial}. Consequently, the convergence rate of \(F_N\) is determined by the largest root of the polynomial equation \eqref{eqn: z_polynomial} excluding the root $z_1$.

To analyze whether there exists a root with a larger absolute value than $z_1$, we examine the number of roots in the region \(|z| < z_1 + \epsilon\) using Rouché's Theorem, where \(\epsilon\) is an infinitesimal positive quantity such that \(0 < \epsilon \ll z_1\). From \eqref{eqn: z_polynomial}, we define:
\begin{align}
g_1(z) &= z^{L+1} + (1-z_1)z_1^L, \\
g_2(z) &= z^L.
\end{align}
For all $|z|=z_1+\epsilon$, we have
\begin{align}
\label{eqn: g_1 approx}
    |g_1(z)| &\leq (z_1+\epsilon)^{L+1} + (1-z_1)z_1^L \nonumber \\
    &\approx z_1^{L+1}+\epsilon (L+1) z_1^{L} +(1-z_1)z_1^L,\\
    |g_2(z)| &=  (z_1+\epsilon)^L \approx z_1^{L}+\epsilon L z_1^{L-1}.
\end{align}
Evaluating the approximation of $|g_2(z)|-|g_1(z)|>0$ for all $|z|=z_1+\epsilon$, it follows that
\begin{align}
    z_1^{L}+\epsilon L z_1^{L-1} -z_1^{L+1}-\epsilon (L+1) z_1^{L} -(1-z_1)z_1^L>0 ,
\end{align}
which can be simplified to 
\begin{align}
\label{eqn: z inequality}
    z_1+\epsilon L > z_1^2+\epsilon (L+1) z_1 +(1-z_1)z_1.
\end{align}
From this, it follows that \(z_1 < \frac{L}{L+1}\). 

Based on the condition \(z_1 < \frac{L}{L+1}\), we discuss the following three cases.

\textbf{Case 1:} $z_1< \frac{L}{L+1}$

Since \(|g_2(z)| > |g_1(z)|\) for all $|z|=z_1+\epsilon$, according to Rouché's Theorem, there is exactly one maximal root in the region $|z|\geq z_1+\epsilon$.
Given there is an odd root in the region $|z|\geq z_1+\epsilon$, this maximal root must be a real number. 

Based on Descartes's rule of signs, if \(L+1\) is even, then there are two positive real roots, \(z_{+}\) and \(z_1\). If \(L+1\) is odd, the polynomial has two positive real roots, \(z_{+}\) and \(z_1\), and one negative real root, \(z_{-}\).

Considering there is exactly one maximal and real root in the region $|z|\geq z_1+\epsilon$, we have $z_+>z_1$, $z_+>z_1 \geq |z_-|$, or $|z_-|>z_1 \geq z_+$. If the absolute value \(|z_-|\) is the maximal root, the function \(g(z_1 + \epsilon)>0\) since $z_1$ is the largest positive real root. However, our previous analysis demonstrated that when \(z_1 < \frac{L}{L+1}\), the function \(g(z_1 + \epsilon)<0\). This is a contradiction.
Therefore, whether \(L+1\) is even or odd, there is a unique maximal positive real root, denoted as \(z_+\), such that \(z_+ > z_1\).

To find the convergence rate in Case 1, we evaluate the function $g(z)$ at the point \(z = 1-\alpha\),
\begin{align}
\begin{aligned}
       g(1-\alpha) &= (1-\alpha)^{L+1} - (1-\alpha)^{L} + \alpha \\
             &= (1-\alpha)^{L} (-\alpha) + \alpha.
\end{aligned}
\end{align}
Since the inequality \(1 > (1-\alpha)^L\) holds, we can conclude $g(1-\alpha)>0$. $g(1-\alpha)>0$ implies that the maximal root satisfies the inequality \(1-\alpha > z_+ > z_1\).

Furthermore, we observe another upper bound for \(z_+\). When \(z = \frac{2L}{L+1} - z_1\), 
\begin{align}
\begin{aligned}
    f(z_1) &:= g(\frac{2L}{L+1} - z_1) \\
    &= \left(\frac{2L}{L+1} - z_1\right)^{L+1} - \left(\frac{2L}{L+1} - z_1\right)^L + (1-z_1)z_1^L.
\end{aligned}
\end{align}
To examine the sign of \( f(z_1) \), we compute its derivative  with respect to $z_1$:
\begin{align}
\label{eqn: g1 2L/L+1 result}
\begin{aligned}
    \frac{d f(z_1)}{dz_1} &= \left(\frac{2L}{L+1} - z_1\right)^{L-1} \left[(L+1)z_1 - L\right] \\
    &\phantom{=}- z_1^{L-1} \left[L - (L+1)z_1\right] \\
    &= \left[z_1^{L-1}-\left(\frac{2L}{L+1} - z_1\right)^{L-1} \right] \left[L - (L+1)z_1\right].
\end{aligned}
\end{align}
Given the condition \( z_1 < \frac{L}{L+1} \), it follows that the term \( \left[L - (L+1)z_1\right] \) is positive. This condition also yields \(\left(\frac{2L}{L+1} - z_1\right) > z_1\). Consequently, the derivative of the function \( f(z_1) \) is always negative. Given that for \( z_1 = \frac{L}{L+1} \), \( f(z_1) = 0 \), the function \( f(z_1) \) remains positive under the condition \( z_1 < \frac{L}{L+1} \).
Thus, $z_+ \leq \frac{2 L}{L+1}-z_1$.

Hence, if \(\cos^2 (\bar\Delta T_s) < \frac{L}{L+1}\), we have one maximal root $z_+$ in the region $(\cos^2 (\bar\Delta T_s)+\epsilon,\min\left\{ \sin^2 (\bar\Delta T_s) [\cos^2 (\bar\Delta T_s)]^L, \frac{2L}{L+1} - \cos^2 (\bar\Delta T_s) \right\})$. Therefore, the convergence rate can be obtained as
    \begin{align}
        \eta = \min\left\{ 1- \alpha, \frac{2L}{L+1} - \cos^2 (\bar\Delta T_s) \right\}.
    \end{align}
\textbf{Case 2:} $z_1 > \frac{L}{L+1}$

In Case 2, we apply Rouché's Theorem to the region $|z|< \frac{2L}{L+1}-z_1$.
When $|z|=\frac{2L}{L+1}-z_1$,
\begin{align}
    &|g_1(z)| \leq (\frac{2L}{L+1}-z_1)^{L+1}+ (1-z_1)z_1^L,\\
    &|g_2(z)| = (\frac{2L}{L+1}-z_1)^{L}.
\end{align}
The condition $|g_1(z)|<|g_2(z)|$ is equivalent to $g(\frac{2L}{L+1}-z_1)<0$. In light of \eqref{eqn: g1 2L/L+1 result}, if $z_1>\frac{L}{L+1}$, $L-(L+1)z_1$ is negative and $z_1>\frac{2L}{L+1}-z_1$. Thus, the derivative of function $f(z_1)$ is always negative. Since $f(\frac{L}{L+1})=0$, $f(z_1)$ is negative under the condition $z_1>\frac{L}{L+1}$. That is, $|g_1(z)|<|g_2(z)|$ for all $|z|=\frac{2L}{L+1}-z_1$.

According to Rouché's Theorem, we have $L$ roots in the region $|z|< \frac{2L}{L+1}-z_1$. Also, the unique root outside of this region is $z_1$. However, since the root $z_1$ has been added to \eqref{eqn: original z trans}, $z_1$ will not determine the convergence rate of $F_k$. 

Therefore, all roots of must lie \eqref{eqn: original z trans} in the region $|z|< \frac{2L}{L+1}-z_1$. Hence, the convergence rate can be found for the case \(\cos^2 (\bar\Delta T_s) > \frac{L}{L+1}\), as
    \begin{align}
        \eta = \frac{2L}{L+1} - \cos^2 (\bar\Delta T_s).
    \end{align}
\textbf{Case 3:} $z_1 = \frac{L}{L+1}$

In Case 3, we let $N \to \infty$ to check whether $z_1$ can be the convergence rate for \eqref{eqn: original z trans}. Substituting $z_1$ to \eqref{eqn: original z trans}, we obtain
\begin{align}
    \begin{aligned}
             z_1 F_N &= (1-z_1)\sum_{l=1}^{L} z_1^{l-1} F_N z^{1-l}\\
             &= (1-z_1)L F_N .
    \end{aligned}
    \end{align}
This equation shows when $z_1$ is the convergence rate, $z_1$ needs to exactly equal $\frac{L}{L+1}$. Thus, the convergence rate is $\eta=\frac{L}{L+1}$.
\end{proof}
In Theorem \ref{cor:convergence_rate}, since all Cases show that the convergence rate is strictly smaller than $1$, as $N$ approaches infinity, the probability of achieving the target state is given by
\begin{align}
    \lim_{N\rightarrow \infty} P_{\textnormal{tar}} (N)=1.
\end{align}
Additionally, given $\frac{2L}{L+1}-\cos^2 (\bar\Delta T_s) < 1-\alpha$, the convergence rate in all cases is \begin{align}
    \eta=\frac{2L}{L+1} - \cos^2 (\bar\Delta T_s).
\end{align}
Thus, the lower bound on the probability of achieving the target state $P_{\textnormal{tar}} (N)$ can be calculated using this convergence rate. 

\section{Numerical examples}
\label{sec: numerical_simulation}
This section illustrates our results with numerical simulations for a one-qubit system. Firstly, we show the finite steps required to achieve the desired target state to verify TOMPC numerically. Subsequently, we demonstrate the robustness of qTOMPC by comparison with TOMPC, and the gradient-ascent impulse engineering (GRAPE) approach. 
\subsection{TOMPC Applied to the Nominal System}
\label{sec: TOMPC_for_nominal_system}
For the nominal system \eqref{eqn: discrete two-level system nominal}, the free Hamiltonian is designated as $ H_0=\omega_0 \sigma_z$, where $\omega_0=50 MHz$. The control Hamiltonian is given by $H_u = u_x(t) \sigma_x + u_y(t) \sigma_y$. Also, the simulation parameters are as follows: the initial quantum state and the desired target state are the eigenstates of $H_0$, $|0\rangle$ and $|1\rangle$, respectively. In the time-optimal control problem \eqref{eqn: general_qMPC}, the prediction horizon is selected as $L=10$, the weight is $\theta=1.9$, and the control bound is $B=0.5 GHz$. Additional parameters are the sample time $T_s=1 ns$ and the total number of iterations $N=100$. Furthermore, for each iteration in the TOMPC, we utilize IPOPT \cite{wright2006numerical} as the optimization solver.

By employing a TOMPC approach, we determine the optimal time required for the nominal system. The result is shown in Fig. \ref{fig: TOMPC}. In Fig.  \ref{fig: TOMPC}, (a) quantifies the measured probability of the current state relative to the target state $|\langle 1|\psi_{k}\rangle|^2$. (b)  depicts the minimal steps required in every iteration. 
\begin{figure}[htbp]
\centering
\includegraphics[width=0.7\columnwidth]{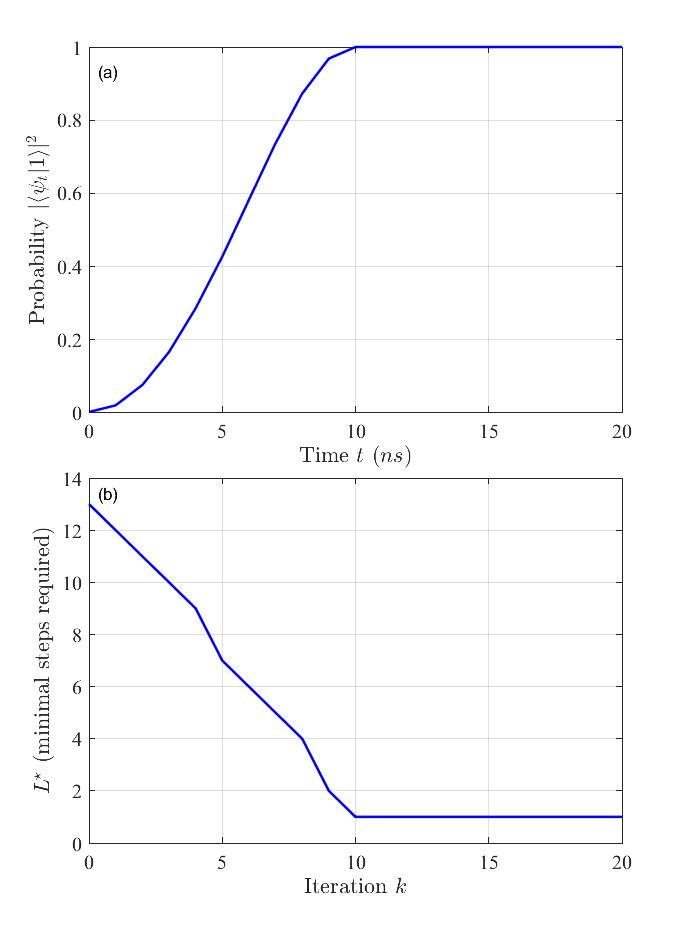}
\caption{(a) the probability $|\langle 1|\psi_{k}\rangle|^2$ of every time $t$. (b) the minimal number of steps required of each iteration.}
\label{fig: TOMPC}
\end{figure}

\subsection{Robustness comparisons}
We compare the robustness of qTOMPC with two other approaches including the TOMPC and the GRAPE methods, where we implement the GRAPE using the QuTiP package  \cite{johansson2012qutip}. 
To assess the efficacy of these three control strategies, we consider the accumulated tracking error and the infidelity between the evolved final state and the desired target state.

In this simulation, we consider the quantum system with uncertainties \eqref{eqn: discrete two-level system}. The parameters are as defined in Subsection \ref{sec: TOMPC_for_nominal_system}, and the uncertain Hamiltonian is formulated as $H_{\Delta} = \Delta_x(k) \sigma_x + \Delta_y(k) \sigma_y$ ($k=0,1, \ldots, N-1$), where $\Delta_x(k)$ and $\Delta_y(k)$ are piece-wise constant scalars representing the discretized components of the uncertain vector $\vec{\Delta}_k$ in \eqref{eqn: discrete two-level system} with a sample time $T_s$.

In the two-level quantum system, we consider uncertainties in three situations: periodic uncertainties, uniform random uncertainties, and Gaussian uncertainties. (i) The periodic uncertainties are characterized by two functions, $\Delta_x (k) = 50 \cos (\omega_x  k T_s +\phi_x) MHz$ and $\Delta_y (k) = 50 \sin (\omega_y k T_s + \phi_y) MHz$ with frequencies $\omega_x, \omega_y \in [15\pi, 25\pi]$ and phase shift $\phi_x, \phi_y \in [-\pi, \pi]$. (ii) In the uniform random uncertainties, $\Delta_x(k), \Delta_y(k) \in[-50 MHz, 50 MHz]$ follow a uniform distribution. (iii) For Gaussian uncertainties, we set $\Delta_x(k), \Delta_y(k)$ to follow a truncated Gaussian distribution over the interval $[-50 MHz, 50 MHz]$ with a mean of $0$ and standard deviation of $25 MHz$.

For the quantum system with uncertainties, we evaluate algorithm performance by calculating accumulated tracking errors over the entire periods $N$, and we repeat simulations $3000$ times for averaged performance.
Here, the accumulated tracking error, denoted by $E_{\text{track}}$, quantifies the deviation from the nominal path and is defined as:
\begin{align}
\label{eqn:error track}
E_{\text {track }}=\sum_{t=0}^{N-1} \left\| |\psi_{1|k}\rangle-|\psi_{1+k}'\rangle \right\|^2,
\end{align}
where $| \psi_{1+k}' \rangle$ is the post-measurement state in Step $4$ of Algorithm \ref{alg:qMPC}.
In the TOMPC and GRAPE approaches, $|\psi_{1|k}\rangle$ represents the evolution of the nominal system. Conversely, the state $|\psi_{1+k}'\rangle$ is obtained by evolving the system including uncertainties.

A comparative evaluation is conducted on the three algorithms, qTOMPC, TOMPC, and GRAPE, considering various uncertainties. Our results demonstrate the qTOMPC algorithm has superior performance in terms of the accumulated tracking error and infidelity, as illustrated in Tables \ref{tab: error_track_close} and \ref{tab: infidelity_close}. Tables \ref{tab: error_track_close} and \ref{tab: infidelity_close} report the average accumulated 
tracking error ($E_{\text{track}}$) and infidelity, respectively, under periodic, uniform random, and Gaussian uncertainties. 
Furthermore, we show the mean of the probability of $|\langle 1|\psi_{1+k}'\rangle|^2$ and the accumulated tracking error in Fig. ~\ref{fig: mean} under the influence of uniform uncertainties. In Fig. ~\ref{fig: median}, we illustrate the median and interquartile range (IQR) for this uncertainties.

The numerical results reveal that the qTOMPC algorithm enhances robustness compared to TOMPC by integrating quantum measurements. Additionally, the qTOMPC leads to a high-fidelity quantum target state and small accumulated tracking error.
\begin{table}[htbp]
\centering
\caption{Average accumulated tracking error for qTOMPC, TOMPC, and GRAPE algorithms under periodic, uniform random, and Gaussian uncertainties.}
\begin{tabular}{|l|l|l|l|l|}
\hline
$E_{\text {track }}$                & qTOMPC & TOMPC & GRAPE \\ \hline
Periodic uncertainties       &  $0.185 $   &   $9.98$   &     $8.39 $        \\ \hline
Uniform random uncertainties &  $0.129$    &  $7.51$   &     $5.69$     \\ \hline
Gaussian uncertainties       &  $7.53 \times 10^{-2}$   &   $4.31$   &     $3.41 $    \\ \hline
\end{tabular}
\label{tab: error_track_close}
\end{table}

\begin{table}[htbp]
\centering
\caption{Infidelity between final and target states for qTOMPC, TOMPC, and GRAPE algorithms under periodic, uniform random, and Gaussian uncertainties.}
\begin{tabular}{|l|l|l|l|l|}
\hline
Infidelity             & qTOMPC & TOMPC & GRAPE \\ \hline
\begin{tabular}[c]{@{}l@{}}Periodic\\uncertainties\end{tabular} & $1.01 \times 10^{-2}$ & $1.91 \times 10^{-1}$ & $1.55 \times 10^{-1}$ \\ \hline
\begin{tabular}[c]{@{}l@{}}Uniform random\\uncertainties\end{tabular} & $6.82 \times 10^{-3}$ & $1.45 \times 10^{-1}$ & $1.09 \times 10^{-1}$ \\ \hline
Gaussian uncertainties       &  $5.72 \times 10^{-3} $   &   $8.37 \times 10^{-2}$   &     $6.54 \times 10^{-2}$    \\ \hline
\end{tabular}
\label{tab: infidelity_close}
\end{table}

\begin{figure}[htbp]
\centering
\includegraphics[width=0.7\columnwidth]{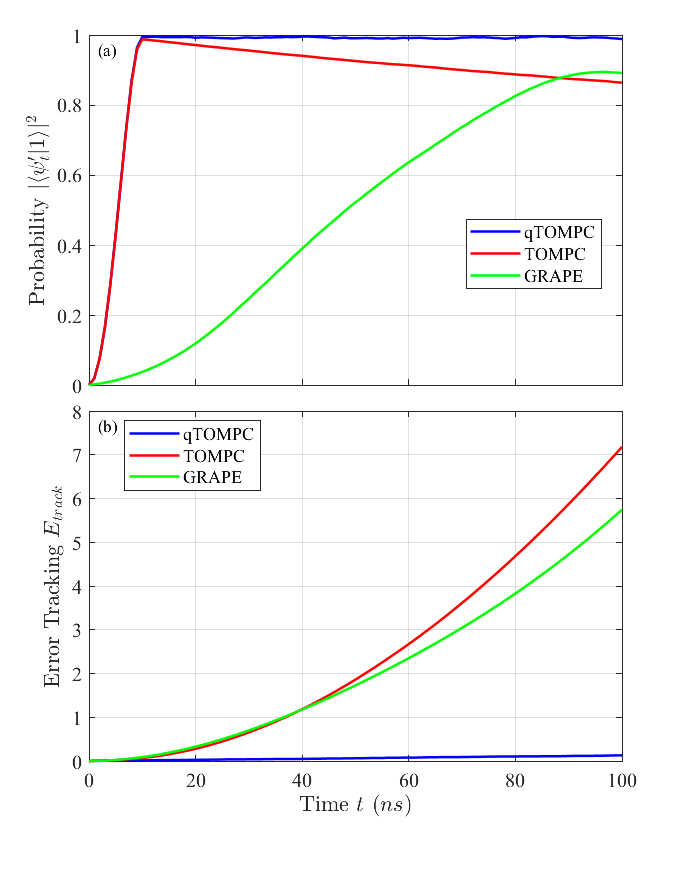}
\caption{Performance of three control strategies for quantum systems with uniform random uncertainties. (a) the probability derived from measurements of the target state $|1\rangle$ over discrete time; (b) a comparative analysis of the accumulated error over time. Specifically, the blue, red, and green lines represent the trajectories associated with the qTOMPC, TOMPC, and GRAPE algorithms, respectively. To assess average behavior, the Monte Carlo method is employed, simulating the system $3000$ times.}
\label{fig: mean}
\end{figure}

\begin{figure}[htbp]
\centering
\includegraphics[width=0.7\columnwidth]{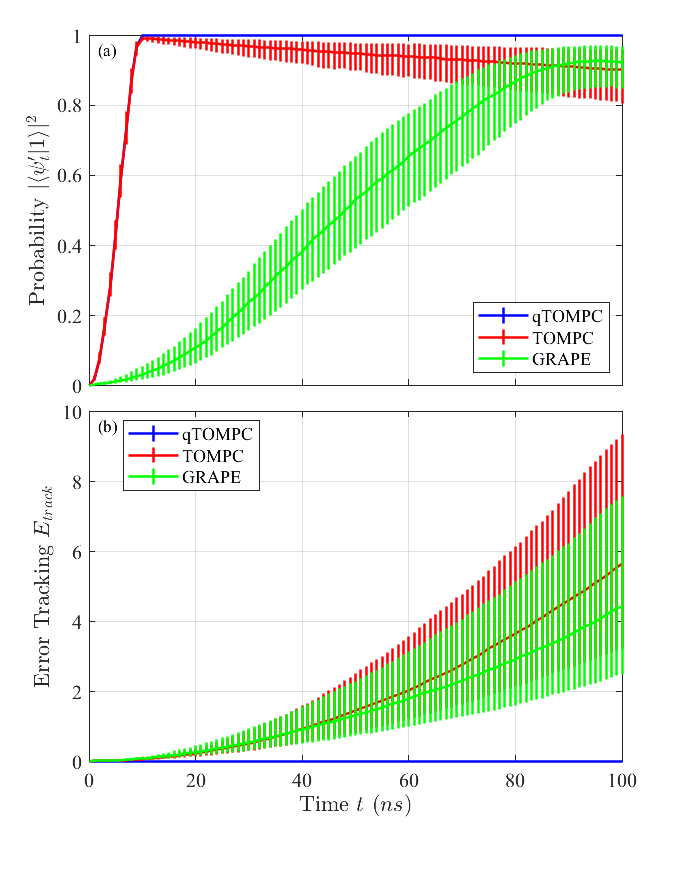}
\caption{This figure parallels Fig. ~\ref{fig: mean} but focuses on the median and interquartile range (IQR). The same simulation conditions, trajectory comparisons, and algorithm references as those in Fig. \ref{fig: mean} are applied.}
\label{fig: median}
\end{figure}

\section{Conclusions}
This paper proposed an MPC scheme to address uncertainties in the system Hamiltonian for two-level quantum systems. The control input sequence is determined through open loop MPC, enabling online control planning and the flexibility to adjust the next quantum trajectory based on the post-measurement quantum state. A sufficient condition for designing a control law that guarantees the desired level of robustness was provided. The proposed qMPC approach has potential applications in robust quantum information processing for quantum systems with uncertainties and can be applied with various measurement methods for effective quantum control based on measurement outcomes. Future work will focus on developing qMPC for multi-level quantum systems and open quantum systems.

\appendix
\section{Proof of Theorem \ref{Thm: TOMPC}}
\label{sec: proof theorem TOMPC}
Given the representation \(|\psi\rangle \langle \psi| = \frac{1}{2}(\mathbb{I}+x_{1} \sigma_x+x_{2} \sigma_y+x_{3} \sigma_z)\), a quantum state \(|\bar\psi_{l|k}\rangle\) corresponds to a Bloch vector \(x_l= [x_{l,1}, x_{l,2}, x_{l,3}]^T \in \mathbb{R}^3\) and we define the initial state \(|\bar\psi_{k}\rangle\) corresponding to Bloch vector $\bar x$. By expressing the quantum state as a Bloch vector, the discrete-time evolution \eqref{eqn: discrete two-level system nominal} can be reformulated as \(x_{l+1}=f_d (x_{l},u_l)\). The function \(f_d(\cdot, \cdot)\) can be determined using the methods described in \cite{yang2013exploring}.
We introduce the objective function:
\begin{align}
    \phi (x_0, \ldots, x_{L^\star-1}) := \frac{1}{2}\sum_{l=0}^{L^{\star}-1} \theta^{l}\| x_l-x_f\|_2,
\end{align}
where \(x_f\) represents $|\psi_f\rangle$ using the Bloch vector. The optimization problem is transformed as:
\begin{align}
\label{eqn: part_TOMPC}
\begin{aligned}
&\min_{u_0,\ldots,u_{L^\star-1}}  && \phi (x_0, \ldots, x_{L^\star-1}) \\ 
&\text{s.t.} 
&& x_0= \bar x, \\
&&& x_{l+1} = f_d (x_l,u_l), \\
&&& |u_{l}| \leq B, \\
&&& x_{L^\star}= x_f+\epsilon, 
\end{aligned}
\end{align}
where \(\epsilon \in \mathbb{R}^{3}\). To incorporate the constraint \(x_{L^\star}=x_f+\epsilon\) into \(\phi\), we employ the Lagrange multiplier \(\lambda_{L^\star}\), which belongs to \(\mathbb{R}^3\). The Lagrange multiplier corresponding to each axis \(i=1,2,3\) is expressed as:
\begin{align}
    \lambda_{L^{\star}, i}(\theta) = \left.\frac{\partial \phi\left(x^{\star}\left(x_{f,i}+\epsilon_i, \theta\right)\right)}{\partial \epsilon_i}\right|_{\epsilon_i=0} ,
\end{align}
where \(x_{f,i}\) indicates the value on the \(x_f\) axis. 
The path \(x^{\star}\left(x_{f,i}+\epsilon_i, \theta\right)=[x_0^\star(x_f+\epsilon_i,\theta),\ldots,x_{L^{\star}}^\star(x_{f,i}+\epsilon_i,\theta)]\), and for \(j \neq i\), \(\epsilon_j=0\).


Following this, we present a proposition regarding the Lagrange multiplier \(\lambda_{N^{\star}}(\theta)\):
\begin{proposition}
\label{prop: theta norm 2}
  For an arbitrary initial state \(x_l\), a threshold \(\theta_1\) exists such that \(\theta^{L^{\star}} \geq \left\|\lambda_{L^{\star}}(\theta)\right\|_{2}\) whenever \(\theta \geq \theta_1\).
\end{proposition}

\begin{proof}
All \( x_k \) lie on the Bloch sphere, which means that the distance between any two such points is bounded. This fact implies the function \( f_d \) is Lipschitz continuous. Furthermore, the function \( f_d \) is known to be differentiable and invertible. Thus, we can deduce the following relationship:
\begin{align}
\label{eqn: bounded disturbed trajectory}
    \left|\frac{\partial x_l^{\star}(x_{f,i}+\epsilon_i, \theta)}{\partial \epsilon_i}\right| \leq \beta, \quad \forall l \in \{0, \ldots, L^{\star}\},
\end{align}
where \( \beta \in \mathbb{R} \).

With the bounded trajectory \eqref{eqn: bounded disturbed trajectory} and using the fact that the quantum states are bounded on the Bloch sphere, the term $|x_l^\star (x_{f,i}+\epsilon_i,\theta)-x_f|$ has a finite value and we can find a bound on the Lagrange multipliers $\lambda_{L^{\star}, i}(\theta)$ as 
\begin{align}
\label{eqn: abs lambda}
\begin{aligned}
\left|\lambda_{L^{\star}, i}(\theta)\right|= & \left|\left.\frac{\mathrm{d} \phi\left(x^{\star}\left(x_{f,i}+\epsilon_i,  \theta\right)\right)}{\mathrm{d} \epsilon_i}\right|_{\epsilon_i=0} \right| \\
= & \left|\frac{\partial \phi\left(x^{\star}(x_{f,i}, \theta)\right)}{\partial x^{\star}(x_{f,i}, \theta)} \cdot \frac{\partial x^{\star}(x_{f,i}, \theta)}{\partial \epsilon_i}\right| \\
\leq & \theta^0 \frac{\left| x^\star_0(x_{f,i},\theta)-x_{f,i} \right|}{2\left\|x^\star_0(x_{f},\theta)-x_f\right\|_2} \left\|\frac{\partial x_0^{\star}(x_{f,i}, \theta)}{\partial \epsilon_i}\right\|_1+\ldots \\
+&\theta^{L^{\star}-1}  \frac{\left| x^\star_{L^\star-1}(x_{f,i},\theta) -x_f \right|}{2\left\|x^\star_0(x_{f},\theta)-x_f\right\|_2} \left\|\frac{\partial x_{L^{\star}-1}^{\star}(x_{f,i}, \theta)}{\partial \epsilon_i}\right\|_1 \\
= & \mathcal{O}\left(\theta^{L^{\star}-1}\right), \quad \theta \rightarrow \infty.
\end{aligned}
\end{align}
From \eqref{eqn: abs lambda}, we deduce:
    \begin{align}
        \begin{aligned}
\left\|\lambda_{L^{\star}}(\theta)\right\|_{2} & = \sqrt{\lambda^2_{L^{\star},1}(\theta)+\lambda^2_{L^{\star},2}(\theta)+\lambda^2_{L^{\star},3}(\theta)}\\
&= \mathcal{O}(\theta^{L^\star-1}) .
\end{aligned}
    \end{align}
Therefore, there exists a \(\theta \geq \theta_1\) such that \(\theta^{L^\star} \geq \left\|\lambda_{L^{\star}}(\theta)\right\|_{2}\).
\end{proof}

Proposition \ref{prop: theta norm 2} provides the condition \(\theta^{L^\star} \geq \left\|\lambda_{L^{\star}}(\theta)\right\|_{2}\), which can show OCP \eqref{eqn: general_qMPC} corresponds to the objective function $\phi$ and adding a penalty function. Specifically, 
we reformulate OCP \eqref{eqn: general_qMPC} using the Bloch vector form \eqref{eqn: part_TOMPC} as
\begin{align}
    \begin{aligned}
    &\min_{u_0,\ldots,u_{L-1}}  && \phi (x_0, \ldots, x_{L^\star-1}) + \frac{1}{2}\sum_{l=L^\star}^{L} \theta^{l} \left\| x_l - x_f\right\|_2 \\ 
    &\text{s.t.} 
    && x_0= \bar x, \\
    &&& x_{l+1} = f_d (x_l,u_l), \\
    &&& |u_{l}| \leq B, \\
    &&& x_{L}= x_f. 
    \end{aligned}
    \end{align}
In the special case where \(L=L^{\star}\), the endpoint constraint is represented as \(x_{L^{\star}}=x_f\), which validates the theorem for this specific scenario. 

When \(L > L^{\star}\) and for any given \(\theta \geq \theta_1\), the term \(\frac{1}{2}\theta^{L^{\star}}\left\|x_{L^{\star}}-x_f\right\|_2\) serves as a penalty that enforces the constraint \(x_{L^{\star}}=\epsilon=x_f\). This holds according to \cite{wright2006numerical}, given the self-duality of the 2-norm and the condition \(\theta^{L^\star}>\left\|\lambda_{L^\star}(\theta)\right\|_2\) from Proposition \ref{prop: theta norm 2}. Furthermore, based on Assumption \ref{ass: target stable}, proving \(x_{L^{\star}}=x_f\) also ensures that \(x_l = x_f\) for \(l = L^{\star} + 1, \ldots, L\).
\section{Proof of $h(T_s) > 0$ for $0 < T_s \leq \frac{\pi}{|\vec v + \vec \Delta|+|\vec \Delta|}$ and $\vec v \nparallel \vec \Delta$}
\label{sec: Proof of Decreasing Series in general case}
We will demonstrate that the Taylor expansion 
\begin{align}
    h(T_s)=\sum_n \frac{C_n}{n !} T_s^n
\end{align}
is an alternating decreasing series, starting from the $4$th-order term. Here, $n$ covers only even integers, and $C_n$ are the coefficients of the Taylor expansion. Based on the properties of an alternating decreasing series, $h(T_s)$ must exceed its truncated Taylor series.

For convenience, we write 
\begin{align}
    a_1=\frac{|\vec{v}| |\vec{v}+\vec{\Delta}|+\vec{v}\cdot(\vec{v}+\vec{\Delta})}{2|\vec{v}| |\vec{v}+\vec{\Delta}|}    
\end{align}
and 
\begin{align}
    a_2=\frac{|\vec{v}| |\vec{v}+\vec{\Delta}|-\vec{v}\cdot(\vec{v}+\vec{\Delta})}{2|\vec{v}| |\vec{v}+\vec{\Delta}|}.   
\end{align}
We also denote $f_1=|\vec{v}+\vec{\Delta}|-|\vec{v}|$ and $f_2=|\vec{v}+\vec{\Delta}|+|\vec{v}|$.
As a result, the function $h(T_s)$ can be written as
\begin{align}
    h(T_s) = a_1 \cos(f_1 T_s) + a_2 \cos(f_2 T_s) - \cos(|\vec\Delta| T_s).
\end{align}
By employing the Cauchy–Schwarz inequality, we ascertain that
\begin{align}
\label{eqn: amp inequality}
     0< a_1, a_2 < 1.
\end{align}
Furthermore, by using the Triangle inequality, it follows that
\begin{align}
\label{eqn: freq inequality}
    f_1<|\vec \Delta|<f_2.
\end{align}

When conducting a direct calculation of the Taylor series expansion up to the 4th-order term, we find that both the zeroth and the second-order terms are zero, whereas the 4th-order term yields $\frac{T_s^4}{6}|\vec v \times \vec \Delta|^2$. In this context, we consider scenarios where $\vec v$ is not parallel to $\vec \Delta$. As such, it is established that the 4th-order term is strictly positive. Following inequality \eqref{eqn: freq inequality}, it follows that the sign of the Taylor series expansion of $h(T_s)$, from the 4th-order term, aligns with the Taylor series expansion of $a_2 \cos(f_2 T_s)$.

The process under consideration directs us to ascertain that the Taylor series expansion of the function $h(T_s)$ can be viewed as an alternating series commencing from the 4th-order term.

To verify $h(T_s)$ is a monotonically decreasing sequence, we show that the ratio of $(n+2)$th-order term to the $n$th-order is less than $1$.

For $n=4$, the difference $|C_4 | (f_1^2+f_2^2+|\vec \Delta|^2)-|C_{6}|$ can be evaluated as follows:
\begin{align}
    \begin{aligned}
        &|C_4 | (f_1^2+f_2^2+|\vec \Delta|^2)-|C_{6}| \\
        &= (a_1 f_1^4 + a_2 f_2^4 -|\vec \Delta|^4) (f_1^2+f_2^2+|\vec \Delta|^2)\\
        &\phantom{==}- (a_1 f_1^6 + a_2 f_2^6 -|\vec \Delta|^6) \\
        &= f_1^2 f_2^2 (a_1 f_1^{2} +a_2 f_2^{2}) + f_1^2 |\vec \Delta|^2 (a_1 f_1^{2} -|\vec \Delta|^{2}) \\
        &\phantom{==} + f_2^2 |\vec \Delta|^2 (a_2 f_2^{2} -|\vec \Delta|^{2}) \\
        &= f_1^2 f_2^2 |\vec \Delta|^{2} - a_2 f_1^2 f_2^2 |\vec \Delta|^{2}
        -  a_1 f_1^2 f_2^2 |\vec \Delta|^{2} \\
        &= f_1^2 f_2^2 |\vec \Delta|^{2} (1-a_1-a_2) =0 .
        \end{aligned}
\end{align}
We now show that the ratio of the $6$th-order term to the $4$th-order term is strictly less than $1$ as follows:
\begin{align}
\label{eqn: inequality Taylor expansion_appendix 4}
\begin{aligned}
      \frac{4!|C_{6}|}{6!|C_{4}|} T_s^2 & = \frac{4!}{6!}(f_2^2 +|\vec \Delta|^2+f_1^2)T_s^2 \\
      &< \frac{3}{30} (|\vec v +\vec \Delta|+|\vec v|)^2 T_s^2 \\
    &\leq \frac{1}{10} \pi^2 < 1,
\end{aligned}
\end{align}
which follows from inequality \eqref{eqn: freq inequality} and the condition of Case 1,  $f_2 T_s \leq \pi$. It follows that the ratio of the $6$th-order term to the $4$th-order term is indeed less than $1$.

When $n \geq 6$, we can show the ratio of the $n+2$th-order term to the $n$th-order term less than $1$ through the inequality $|C_n|(f_2^2+|\vec\Delta|^2)-|C_{n+2}|>0$. To demonstrate this inequality, firstly, we use the formulas $a_1+a_2=1$ and $a_1 f_1^2+a_2 f_2^2=|\vec \Delta|^2$ to obtain the following relation:
\begin{align}
\label{eqn: pre proof decreasing}
\begin{aligned}
    |\vec\Delta|^n-a_1 f_1^n&=|\vec\Delta|^n- (1-a_2) f_1^n \\
    &=|\vec\Delta|^n- f_1^n + a_2 f_1^n\\
    &=(|\vec\Delta|^2- f_1^2) \sum_{k=0}^{n-2}|\vec\Delta|^{n-2-k} f_1^{k} + a_2 f_1^n\\
    &=a_2(f_2^2- f_1^2) \sum_{k=0}^{n-2}|\vec\Delta|^{n-2-k} f_1^{k} + a_2 f_1^n,
\end{aligned}
\end{align}
where $k$ is an even integer within the interval from $0$ to $n-2$. The term $a_2 (f_2^2-f_1^2)=|\vec \Delta|^2-f_1^2$ in the last part of \eqref{eqn: pre proof decreasing} is inferred from the formula $(1-a_2) f_1^2+a_2 f_2^2=|\vec \Delta|^2$.

Subsequently, we use \eqref{eqn: pre proof decreasing} to establish the inequality $|C_n|(f_2^2+|\vec\Delta|^2)-|C_{n+2}|>0$ as follows:
\begin{small}
\begin{align}
\label{eqn: proof decreasing 1}
\begin{aligned}
      &|C_n|(f_2^2+|\vec\Delta|^2)-|C_{n+2}| \\
      &= (a_1 f_1^n + a_2 f_2^n -|\vec \Delta|^n) (f_2^2+|\vec\Delta|^{2}) \\
      &\phantom{=}- (a_1 f_1^{n+2} + a_2 f_2^{n+2} -|\vec \Delta|^{n+2}) \\
      & = a_2 (f_2^2+|\vec\Delta|^{2}) (f_2^n-f_1^n-(f_2^2- f_1^2) \sum_{k=0}^{n-2}|\vec\Delta|^{n-2-k} f_1^{k})\\
      &\phantom{=}-a_2 (f_2^{n+2}-f_1^{n+2}-(f_2^2- f_1^2) \sum_{k=0}^{n}|\vec\Delta|^{n-k} f_1^{k})\\
      & = a_2 (f_2^2+|\vec\Delta|^{2}) (f_2^2-f_1^2)\left(\sum_{k=0}^{n-2}f_2^{n-2-k} f_1^{k}-\sum_{k=0}^{n-2}|\vec\Delta|^{n-2-k} f_1^{k}\right)\\
      &\phantom{=}-a_2 (f_2^2-f_1^2)\left(\sum_{k=0}^{n}f_2^{n-k} f_1^{k}- \sum_{k=0}^{n}|\vec\Delta|^{n-k} f_1^{k}\right) .
\end{aligned}
\end{align}
\end{small}
Since we know that both $a_2$ and $f_2^2-f_1^2$ are positive, we can extract them from our analysis.
Consequently, we find that:
\begin{align}
\label{eqn: proof decreasing 3}
\begin{aligned}
    &\frac{1}{a_2(f_2^2-f_1^2)}(|C_n|(f_2^2+|\vec\Delta|^2)-|C_{n+2}|) \\
      & = (f_2^2+|\vec\Delta|^{2}) \sum_{k=0}^{n-2} f_1^{k} (f_2^{n-2-k} -|\vec\Delta|^{n-2-k}) \\
      &\phantom{=}-\sum_{k=0}^{n} f_1^{k} (f_2^{n-k} -|\vec\Delta|^{n-k}) \\
    &= \sum_{k=0}^{n-4} f_1^{k} (f_2^2+|\vec\Delta|^{2})(f_2^{n-2-k} -|\vec\Delta|^{n-2-k}) \\
     &\phantom{==}-\sum_{k=0}^{n-4} f_1^{k} (f_2^{n-k} -|\vec\Delta|^{n-k})-f_1^{n-2} (f_2^2-|\vec\Delta|^2)
\end{aligned}
\end{align}
where in the last equality we use the fact that $f_2^0-|\vec\Delta|^{0}=0$. 


Additionally, we use the following relation:
\begin{align}
    \begin{aligned}
        &(f_2^2+|\vec\Delta|^{2})(f_2^{n-2} -|\vec\Delta|^{n-2})-(f_2^{n} -|\vec\Delta|^{n})\\
        &= f_2^2 |\vec\Delta|^2 (f_2^2-|\vec\Delta|^2) \sum_{l=0}^{n-6} f_2^{n-6-l} |\vec\Delta|^l ,
    \end{aligned}
\end{align}
where the index $l$ is an even integer. This allows us to simplify \eqref{eqn: proof decreasing 3} as follows:
\begin{small}
\begin{align}
\label{eqn: proof decreasing 35}
\begin{aligned}
    &\frac{1}{a_2(f_2^2-f_1^2)}(|C_n|(f_2^2+|\vec\Delta|^2)-|C_{n+2}|) \\
     & =(f_2^2-|\vec\Delta|^2)\left[f_2^2 |\vec\Delta|^2\sum_{k=0}^{n-6} f_1^{k}  \sum_{l=0}^{n-6-k}f_2^{n-6-k-l} |\vec\Delta|^{l}
       - f_1^{n-2}   \right].
\end{aligned}
\end{align}
\end{small}
By substituting $n=6$ into \eqref{eqn: proof decreasing 35}, it yields
\begin{align}
    \begin{aligned}
        &\frac{1}{a_2(f_2^2-f_1^2)}(|C_6 | (f_2^2+|\vec \Delta|^2)-|C_{8}|) \\
        & =(f_2^2-|\vec\Delta|^2)\left[f_2^2 |\vec\Delta|^2 - f_1^{4}   \right].  
    \end{aligned}
\end{align}
Given the inequality in \eqref{eqn: freq inequality}, it follows that $|C_6 | (f_2^2+|\vec \Delta|^2)-|C_{8}|>0$.

When $n\geq 8$, dividing \eqref{eqn: proof decreasing 35} by $f_2^2-|\vec\Delta|^2$, we have:
\begin{small}
\begin{align}
\label{eqn: proof decreasing 4}
\begin{aligned}
&\frac{1}{a_2(f_2^2-f_1^2)(f_2^2-|\vec\Delta|^2)}(|C_n|(f_2^2+|\vec\Delta|^2)-|C_{n+2}|) \\
     &=f_2^2 |\vec\Delta|^2\sum_{k=0}^{n-8} f_1^{k}  \sum_{l=0}^{n-8-k}f_2^{n-8-k-l} |\vec\Delta|^{l} +f_2^2 |\vec\Delta|^2 f_1^{n-6}
       - f_1^{n-2} \\
    &= f_2^2 |\vec\Delta|^2\sum_{k=0}^{n-8} f_1^{k}  \sum_{l=0}^{n-8-k}f_2^{n-8-k-l} |\vec\Delta|^{l} + f_1^{n-6} (f_2^2 |\vec\Delta|^2-f_1^4)\\
    &>0.
\end{aligned}
\end{align}    
\end{small}
In \eqref{eqn: proof decreasing 4}, the term corresponding to $k=n-6$ is extracted from the summation. Again, it follows from the inequality $f_2^2 |\vec\Delta|^2-f_1^4>0$ that $|C_n | (f_2^2+|\vec \Delta|^2)-|C_{n+2}|>0$ holds for $n\geq 8$.

Therefore, for $n \geq 6$, we have proved the inequality $|C_n|(f_2^2+|\vec\Delta|^2)-|C_{n+2}|>0$.
Given this result, we can demonstrate that the Taylor series of $h(T_s)$ is decreasing by calculating the ratio of the $(n+2)$th-order term to the $n$th-order term, 
\begin{align}
\label{eqn: inequality Taylor expansion_appendix}
\begin{aligned}
      \frac{n!|C_{n+2}|}{(n+2)!|C_{n}|} T_s^2 & < \frac{n!}{(n+2)!}(f_2^2 +|\vec \Delta|^2)T_s^2 \\
      &< \frac{2}{(n+2)(n+1)} (|\vec v +\vec \Delta|+|\vec v|)^2 T_s^2 \\
    &< \frac{2}{(n+2)(n+1)} \pi^2 < 1 ,
\end{aligned}
\end{align}
which can be obtained using the same reasoning as in \eqref{eqn: inequality Taylor expansion_appendix 4}.

Given that we have established for $n\geq 4$, the ratio of the $(n+2)$th-order term to the $n$th-order term is less than $1$, we can deduce that the coefficients in the Taylor series of $h(T_s)$ is a decreasing series. Therefore, in the case, $0< T_s \leq \frac{\pi}{|\vec v +\vec \Delta|+|\vec \Delta|}$ and $\vec v \nparallel \vec \Delta$, the Taylor series of $h(T_s)$ forms an alternating decreasing series commencing from the $4$th-order term.

Also, the truncated value of $h(T_s)$ up to the sixth term can be calculated as
\begin{align}
\label{eqn: Taylor sixth term for general case}
\begin{aligned}
        \frac{C_4}{4!} T_s^4+\frac{C_6}{6!} T_s^6=\frac{T_s^4}{6}|(\Vec{v} \times \Vec{\Delta})|^2(1-\epsilon_6 T_s^2),
\end{aligned}
\end{align}
where $\epsilon_6$ can be found in the first equality of \eqref{eqn: inequality Taylor expansion_appendix 4}.
\begin{align}
\begin{aligned}
      \epsilon_6=\frac{1}{15}(&\frac{1}{2}|\vec\Delta|^2+|\vec{v}|^2+|\vec{v}+\vec\Delta|^2).
\end{aligned}
\end{align}
Given that the Taylor expansion of $h(T_s)$ is an alternating decreasing series starting from the $4$th-order term, and considering that this fourth-order term is positive, it follows that $h(T_s)$ must exceed its truncated value up to the sixth term
\begin{align}
\begin{aligned}
        h(T_s) >\frac{T_s^4}{6}|(\Vec{v} \times \Vec{\Delta})|^2(1-\epsilon_6 T_s^2).
\end{aligned}
\end{align}
Given the condition $T_s<\frac{\pi}{|\vec{v}+\vec{\Delta}|+ |\vec{v}|}$, it can be deduced that $T_s^2  (|\vec{v}|^2+|\vec{v}+\vec\Delta|^2)<T_s^2 (|\vec{v}+\vec{\Delta}|+ |\vec{v}|)^2<\pi^2$ and $T_s |\vec\Delta| <\pi/2$. As a result, we can infer that $h(T_s)>0$ for $0< T_s \leq \frac{\pi}{|\vec v +\vec \Delta|+|\vec \Delta|}$ and $\vec v \nparallel \vec \Delta$.

\bibliographystyle{IEEEtran}
\bibliography{references}
\end{document}